\documentclass[transaction]{IEEEtran}
\usepackage{bbm}
\usepackage{mathrsfs}
\usepackage{amsthm}
\usepackage{amsfonts}
\usepackage{color}
\usepackage{graphicx}
\usepackage{epsfig,latexsym}
\usepackage{float}
\usepackage{indentfirst}
\usepackage{amsmath}
\usepackage{amssymb}
\usepackage{times}
\usepackage{subfigure}
\usepackage{psfrag}
\usepackage{cite}
\usepackage{bm}
\usepackage{algorithmic,algorithm}
\usepackage{stfloats}
\usepackage{ulem}
\usepackage{lipsum}
\usepackage{multicol}

\newtheorem{theorem}{$\mathbf{Theorem}$}
\newtheorem{lemma}[theorem]{$\mathbf{Lemma}$}

\newtheorem{Definition}{Definition}

\begin{document}

\title{Channel Matrix Sparsity with Imperfect Channel State Information in Cloud-Radio Access Networks}
\author{\IEEEauthorblockN{Di Chen, Zhongyuan Zhao, Zhendong Mao, and Mugen Peng}
\thanks{This work was supported in part by the State Major Science and Technology Special Project under Grant 2016ZX03001020-006, in part by the National Natural Science Foundation of China under Grant 61222103 and Grant 61361166005, in part by the Science and Technology Development Project of Beijing Municipal Education Commission of China (Grant No. KZ201511232036), and in part by the National Program for Special Support of Eminent Professionals.. (Corresponding author: Mugen Peng.)}
\thanks{Copyright (c) 2017 IEEE. Personal use of this material is permitted. However, permission to use this material for any other purposes must be obtained from the IEEE by sending a request to pubs-permissions@ieee.org.}\thanks{Di Chen (e-mail: chendi@bupy.edu.cn), Zhongyuan Zhao (e-mail: zyzhao@bupt.edu.cn), Zhendong Mao (e-mail: mzd@bupt.edu.cn), and Mugen~Peng (e-mail: pmg@bupt.edu.cn) are with the Key Laboratory of Universal Wireless Communications (Ministry of Education), Beijing University of Posts and Telecommunications, Beijing, China.}}
\maketitle

\begin{abstract}
Channel matrix sparsification is considered as a promising approach to reduce the progressing complexity in large-scale cloud-radio access networks (C-RANs) based on ideal channel condition assumption. In this paper, the research of channel sparsification is extend to practical scenarios, in which the perfect channel state information (CSI) is not available. First, a tractable lower bound of signal-to-interference-plus-noise ratio (SINR) fidelity, which is defined as a ratio of SINRs with and without channel sparsification, is derived to evaluate the impact of channel estimation error. Based on the theoretical results, a Dinkelbach-based algorithm is proposed to achieve the global optimal performance of channel matrix sparsification based on the criterion of distance. Finally, all these results are extended to a more challenging scenario with pilot contamination. Finally, simulation results are shown to evaluate the performance of channel matrix sparsification with imperfect CSIs and verify our analytical results.
\end{abstract}


\IEEEpeerreviewmaketitle

\section{Introduction}

The paradigm of cloud-radio access networks (C-RANs) is considered as one of the most promising approaches to satisfy the increasing traffic requirements with low cost and high efficiency [1]. Unlike the conventional cellular base stations, the baseband units (BBUs) are separated with the remote radio heads (RRHs), and a centralized BBU pool can be formed, which can provide great convenient for coordinated signal processing and flexible interference management [2].
\par
Due to dense coverage of RRHs, high computational complexity and unsatisfactory delay experience are caused by full-centralized cooperation processing [3]. In particular, the complexity of the optimal linear transceiver design grows quadratically as the network scale increases, and the power consumption will increase rapidly, which make the full-scale joint RRH cooperation and processing impractical. Therefore, low-complexity cooperation strategies have been widely studied. In particular, to maximize the throughput, distributive channel estimation schemes have been proposed in [4]-[6], and the distributive cooperation of RRHs is studied in [7]-[9]. Moreover, to improve the energy efficiency and reduce the overhead costs of signaling, the coordinated beamforming optimization with distributive cooperation of RRHs was investigated in [10]-[12], in [13], the signal processing algorithms to minimize power consumption is explored, and the corresponding performance analysis was provided in [14]. The joint optimization of energy efficiency and quality of service (QoS) is considered in [15]. In [16]-[18], interference management schemes in C-RANs were studied.
\par
Recently, a novel dynamic nested clustering (DNC) algorithm has been proposed by Zhang \textit{et al.} in [19], which can keep a balance between performance and efficiency of centralized processing. The key idea of DNC algorithm is to reduce the dimension of channel matrices by sparification. In particular, the near-sparsity of channel matrix in large C-RANs is explored, where the links with low channel gains can be ignored due to their weak impact. Therefore, the channel matrix could be approximately regarded as a sparse matrix in large-scale C-RANs. Then the RRHs are divided into disjoint clusters, and the partial centralized processing strategies can be implemented. Moreover, the study results show that the optimal linear precoding matrix can be transformed as a doubly bordered block diagonal (DBBD) matrix, whose blocks can be processed separately in parallel to greatly reduce the complexity of signal detection in the BBU pool.
\par
Although the channel sparsification-based processing scheme proposed in [19] can reduce the computational complexity, its performance cannot be guaranteed in practical systems since it is sensitive to the accuracy of channel state information (CSI). Moreover, the scale of clusters should be strictly constrained due to the performance loss caused by channel estimation errors, while a grand cluster is encouraged to form based on ideal CSIs assumption in [19]. Motivated by that, the channel sparification with imperfect CSIs in C-RANs is studied in this paper, and our main contributions can be summarized as following:
\begin{itemize}
  \item The SINR fidelity with imperfect CSIs is researched, which is defined as a ratio of SINRs with and without channel sparsification, which can characterize the performance loss. In particular, a tractable lower bound of SINR fidelity is derived. Different from the scenario with ideal CSIs in [19], our analytical results show that the cluster scale is constrained to avoid the performance loss caused by imperfect CSIs in large-scale cooperation.
  \item Based on the analytical result, a Dinkelbach-based algorithm is proposed to optimized the sparsified channel matrix and cluster formation.
  \item All these research results are extended into a challenging and practical scenario with pilot contamination, and the simulation results are provided to evaluate the performance gains of channel matrix sparification and verify our analytical results.
\end{itemize}

The rest of this paper is organized as follows. Section II introduces the system model. The lower bound of SINR fidelity and optimization of channel sparsification are given in Section III. Section IV extends our study results to a more challenging scenario with pilot contamination. The simulation results are provided in Section V, and the paper is concluded in Section VI.

\section{System Model}

Consider the uplink transmission of a large C-RAN, in which $N$ RRHs and $K$ user equipments (UEs) are randomly located in a disc-region with $r$ radius. Note that all RRHs and UEs are equipped with a single antenna. The centralized signal progressing at the BBU pool requires global CSI, which can be obtained by pilot-based channel estimation. The observation of pilots at the RRHs can be written as
\begin{eqnarray}
\begin{split}
\mathbf{S} =&[\mathbf{s}_1^H,\cdots,\mathbf{s}_N^H]^H=[\mathbf{h}_1,\cdots, \mathbf{h}_K][\mathbf{\phi}_1^H,\cdots,\mathbf{\phi}_K^H]^H\\
&+[\mathbf{n}_1^H,\cdots,\mathbf{n}_N^H]^H=\mathbf{H\Phi+N},
\end{split}
\end{eqnarray}
where $\mathbf{H}=[\mathbf{h}_1,\cdots, \mathbf{h}_K]$ denotes an $N \times K$ channel matrix, and $\mathbf{h}_i$ denotes the channel vector from UE $i$ to all RRHs. $\mathbf{\Phi}=[\mathbf{\phi}_1^H,\cdots$, and $\mathbf{\phi}_K^H]^H$ denotes an $K \times \tau$ training matrix, $\mathbf{\phi}_i$ denotes a $K \times 1$ training sequence from UE $i$, and $\mathbf{S}=[\mathbf{s}_1^H,\cdots,\mathbf{s}_N^H]^H$ is an $N \times \tau$ observation. $\mathbf{s}_i$ denotes a $\tau \times 1$ observation vector at RRH $i$, and $\mathbf{N}$ is the additive white Gaussian noise at the RRHs. To ensure the pilot contamination can be avoided, $\mathbf{\phi}_i$-s follows the constraints of $\mathbf{\phi}_i^H \mathbf{\phi}_j=0,$ $i \neq j$, and $\tau \geq K$. Moreover, the observations of data blocks at all RRHs can be expressed as
\begin{eqnarray}
\begin{split}
\mathbf{y}=&[y_1,\cdots,y_N]=[\mathbf{h}_1,\cdots, \mathbf{h}_K]diag\{\sqrt{P_1}, \dots, \sqrt{P_K}\}\\ & \cdot [x_1,\cdots,x_K]^H
+\mathbf{n}=  \mathbf{H \sqrt{P} x+n},
\end{split}
\end{eqnarray}
where $\mathbf{P}$ is the transmit power matrix of data block, whose diagonal elements are the transmit power of all users, and $\mathbf{x}=[x_1,\cdots,x_K]^H$ is an $K \times 1$ data block vector. The $i$th element of $\mathbf{x}$ is the data symbol sent by UE $i$, i.e., $\mathbb{E} \{\mathbf{x} \mathbf{x}^\mathit{H}\} =\mathbf{I}_K$. $\mathbf{n}$ denotes the additive white Gaussian noise at RRHs, $\mathbf{n} \sim \it{CN}(\mathbf{0},N_0\mathbf{I}_\tau)$.

The channels are modeled by including both path loss and flat Rayleigh fading at the channel matrix, and $\mathbf{H}$ can be expressed as:

\begin{eqnarray}
\mathbf{H}=[d_{n,k}^{-{\alpha}/{2}}]_{N \times K} \odot [\gamma_{n,k}]_{N\times K}  =\mathbf{D \odot \Gamma},
\end{eqnarray}
where $\mathbf{A \odot B}$ denotes Hadamard product of matrices $\mathbf{A}$ and $\mathbf{B}$. $\mathbf{D}$ and $\mathbf{\Gamma}$ are the matrices characterizing the path loss and the flat Rayleigh fading of channels, respectively, whose elements, i.e., $d_{n,k}^{-\alpha/2}$ and $\gamma_{n,k}$ denote the path loss and the flat Rayleigh fading of the link from UE $k$ to RRH $n$. $\alpha$ denotes the path loss exponent.

Although full centralized processing can maximize the cooperation gains, it may cause high computational complexity which has a significant impact on the QoS guaranteeing in C-RANs. To make a balance between the performance and the efficiency, a channel sparsification strategy with perfect CSIs has been proposed in [19]. To extend the proposed scheme to the practical scenarios, channel sparsification with imperfect CSIs is discussed in this paper.

\subsection{Channel Estimation and Sparsification}

In this paper, LS estimation is used to obtain the channel matrix. To ensure the transmission performance, orthogonal resource units are allocated to UEs, and thus the co-channel interference can be removed completely. The estimates of channel matrix at the BBU pool can be written as
\begin{eqnarray}
\mathbf{H}^{est}= \mathbf{S \Phi}^\dag= \mathbf{H}+ \mathbf{N \Phi}^\dag=  \mathbf{H}+\mathbf{E},
\end{eqnarray}
where $\mathbf{\Phi}^\dag$ is the pseudo inverse matrix of $\mathbf{\Phi}$ with $\,$ $\mathbf{\Phi}^\dag=(\mathbf{ \Phi \Phi}^H)^{-1} \mathbf{\Phi}^H$, and $\mathbf{E}=\mathbf{N \Phi}^\dag$ denotes an $N \times K$ estimation error matrix.
Referring [19], the impact of sparsification is mainly determined by the path loss, i.e., the error caused by the other UEs far from the UE $k$ can be ignored. Therefore, a pathloss-based criterion is used for channel sparsification in this paper, which can be expressed as
\begin{eqnarray}
\widehat{h}_{n,k}=\left\{
\begin{array}{rcl}
\bar{h}_{n,k}+\bar{e}_{n,k} & & {d_{n,k} \leq d_0} \\
0                           & & {else}
\end{array} \right.,
\end{eqnarray}
where $\widehat{h}_{n,k}$ denotes the $\{n,k\}$th entry of sparse observed channel matrix $\widehat{\mathbf{H}}$; $\bar{h}_{n,k}$ and $\bar{e}_{n,k}$ are the $\{n,k\}$th entry of sparse channel matrix $\bar{\mathbf{H}}$ and sparse error matrix $\bar{\mathbf{E}}$, respectively; $d_0$ is the threshold of link distance. The user position can be determined in wireless communication systems by several well-known technologies, such as time difference of arrival (TDOA) [20]. Therefore, the link distance information can be obtained at real C-RANs in practice. Moreover, since the link distance varies relatively slowly comparing with the fast fading coefficients, the amount of overhead to estimate the link distances is far less than the overhead to estimate the fast channel fading coefficients. In this paper, it is assumed to be previously estimated. Since the acquisition of CSIs is challenging in C-RANs due to the large number of channel parameters and time delay among different nodes [21], with the distance-based channel sparsification, the channel estimation overhead can be significantly reduced. As illustrated in Fig.$\,$1, the entries $\bar{h}_{n,k}$ and $\bar{e}_{n,k}$ are set to zero when the link distance $d_{n,k}$ is larger than distance threshold $d_0$. Thus the observed channel can be expressed by
\begin{eqnarray}
\mathbf{\widehat{H}}
=(\mathbf{H}^{est})^{spar}=\mathbf{\bar{H}}+\mathbf{\bar{E}}.
\end{eqnarray}
\begin{figure}[t!]\centering
\label{sys}
\centering
\includegraphics[width=0.5\textwidth]{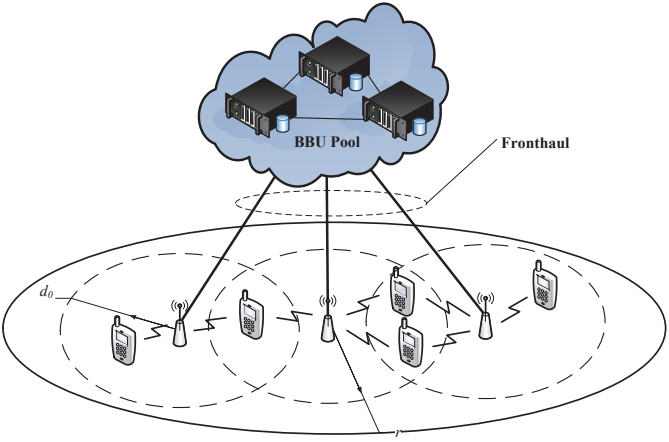}
\vspace*{-10pt} \caption{Link Distance based Clustering in large C-RANs}\vspace*{-10pt} \label{sys}
\end{figure}
Although the channel sparsification can reduce the amount of required CSIs at the BBU pool, the accuracy of CSIs is lowered by sparsification error.
After sparsification, the actual channel matrix $\mathbf{H}$ is divided into two parts, denoted by $\mathbf{H=\bar{H}+\widetilde{H}}$, where $\mathbf{\widetilde{H}}$ consists of channel coefficients with link distance greater than $d_0$. The estimation error matrix can be denoted by $\mathbf{E=\bar{E}+\widetilde{E}}$.
Thus, the actual channel experienced by signal can be derived as
\begin{eqnarray}
\mathbf{H}=  \mathbf{\widehat{H}} + \mathbf{\widetilde{H}}- \mathbf{\bar{E}},
\end{eqnarray}
where the second term is the sparsification error and the third term is the estimation error, respectively.

\subsection{Data Detection}
Based on (2) and (7), the received data signal $\mathbf{y}$ can be rewritten as
\begin{eqnarray}
\mathbf{y} = \mathbf{\widehat{H}\sqrt{P}x} +  \mathbf{\widetilde{H}\sqrt{P}x} -  \mathbf{\bar{E}\sqrt{P}x} + \mathbf{n}.
\end{eqnarray}

By using MMSE detector, the detection of data vector $\mathbf{x}$ can be expressed by
\begin{eqnarray}
\mathbf{\widehat{x}}=\mathbf{\widehat{V}}^H \mathbf{y},
\end{eqnarray}
and the $k$th element in $\mathbf{x}$ can be expressed as
\begin{eqnarray}
\begin{split}
\widehat{x}_k =& \sqrt{P_k} \mathbf{v}_k^H \mathbf{\widehat{h}}_k x_k + \sqrt{P_k} \mathbf{v}_k^H (\mathbf{\widetilde{h}}_k -\mathbf{\bar{e}}_k) x_k \\
&+  \mathbf{v}_k^H \sum_{j\neq k}\sqrt{P_j}\mathbf{h}_j x_j + \mathbf{v}_k^H \mathbf{n},
\end{split}
\end{eqnarray}
where the second one is the interference caused by the same transmitter, the third one is the interference resulting from other transmitters, and the last one is caused by noise in receiver.
\par
The MMSE beamforming matrix in (9) is generated by
\begin{eqnarray}
\mathbf{\widehat{V}} = \mathbf{\widehat{A}}^{-1} \sqrt{\mathbf{P}} \mathbf{\widehat{H}},
\end{eqnarray}
where
\begin{eqnarray}
\mathbf{\widehat{A}} = \mathbf{\widehat{H}}\mathbf{P} \mathbf{\widehat{H}}^H + \mathbf{\Gamma} + N_0 \mathbf{I},
\end{eqnarray}
and
\begin{eqnarray}
\begin{split}
\mathbf{\Gamma} =&
\mathbb{E} \left\{ (\mathbf{\widetilde{H}}-\mathbf{\bar{E}}) \mathbf{P} \mathbf{\widehat{H}}^H+\mathbf{\widehat{H}} \mathbf{P} (\mathbf{\widetilde{H}}-\mathbf{\bar{E}})^H \right. \\
&\left.+(\mathbf{\widetilde{H}}-\mathbf{\bar{E}}) \mathbf{P}(\mathbf{\widetilde{H}}-\mathbf{\bar{E}})^H \right\}.
\end{split}
\end{eqnarray}

Conditioning on the distance threshold $d_0$, the matrices $\mathbf{\bar{H}}, \mathbf{\widetilde{H}}, \mathbf{\bar{E}}$ and $\mathbf{\widetilde{E}}$
are mutually independent. With randomly located RRHs and UEs, $\mathbb{E} \lbrace \widetilde{\mathbf{H}}\widetilde{\mathbf{H}}^H
 \rbrace$ and $ \mathbb{E} \lbrace \bar{\mathbf{E}}\bar{\mathbf{E}}^H \rbrace$ turn to diagonal matrices with unified diagonal elements, and other terms in the right hand of (13) turn to null matrices. Thus, (13) can be written as
\begin{eqnarray}
\mathbf{\Gamma}=N_1 \mathbf{I}+N_2\mathbf{I},
\end{eqnarray}
where $N_1=\mathbb{E} \left\{ \sum_{j=1}^K P_j \vert \widetilde{h}_{nj} \vert^2 \right\}
  $ and $N_2= \mathbb{E} \left\{ \sum_{j=1}^K P_j \vert \bar{e}_{nj} \vert^2 \right\} $ for the arbitrary RRH $n$, which characterize the ignoring actual channels and the remaining error after sparsification, respectively.
Notice that the number of RRHs should be greater or at least close to the number of UEs, since during the progress of $\mathbb{E} \lbrace \widetilde{\mathbf{H}}\widetilde{\mathbf{H}}^H
 \rbrace$ and $ \mathbb{E} \lbrace \bar{\mathbf{E}}\bar{\mathbf{E}}^H \rbrace$, the number of RRHs is assumed to be large enough to apply Law of Large Numbers.
Treating the first term in (10) to be the signal, the other three terms
 as the interference plus noise, the $SINR$ of user $k$ becomes
\begin{eqnarray}
\begin{split}
 &\widehat{SINR}_k(d_0) =\\
&\frac{P_k|\mathbf{\widehat{v}}_k^H \mathbf{\widehat{h}}_k|^2}{P_k|\mathbf{\widehat{v}}_k^H(\mathbf{\widetilde{h}}_k - \mathbf{\bar{e}}_k)|^2
 +\sum_{j \neq k}P_j|\mathbf{\widehat{v}}_k^H \mathbf{h}_j|^2 + N_0\mathbf{\widehat{v}}_k^H \mathbf{\widehat{v}}_k}.
 \end{split}
\end{eqnarray}
\par To investigate the relative performance loss, an ideal $SINR_k$ without channel sparsity is needed. Likewise, it can be expressed as
\begin{eqnarray}
SINR_k = \frac{ P_k|\mathbf{v}_k^H \mathbf{h}_k|^2}{\sum_{j \neq k} P_j |\mathbf{v}_k^H \mathbf{h}_j|^2 + N_0 \mathbf{v}_k^H \mathbf{v}_k},
\end{eqnarray}
where $\mathbf{h}_k$ and $\mathbf{v}_k$ being the $k$th column of actual channel matrix $\mathbf{H}$ and corresponding detection matrix $\mathbf{V}$. The detection matrix $\mathbf{V}$ can be expressed as

\begin{eqnarray}
  \mathbf{V} =  \mathbf{A}^{-1}\sqrt{\mathbf{P}}\mathbf{H} ,
  \end{eqnarray}
where
  \begin{eqnarray}
  \mathbf{A}= \mathbf{H P H}^H + N_0 \mathbf{I}.
  \end{eqnarray}

Both (17) and (18) are derived based on MMSE detection. Next, the SINR fidelity is analyzed in the following section.
\newcounter{mytempeqncnt}
\begin{figure*}[!t]
\normalsize
\setcounter{mytempeqncnt}{\value{equation}}
\setcounter{equation}{19}
\begin{eqnarray}
\label{eqn_dbl_x}
\mathbb{E}\lbrace \widehat{SINR}_k(d_0) \rbrace
=\mathbb{E} \left\{ \frac{ P_k \vert \mathbf{\widehat{v}}_k^H \mathbf{\widehat{h}}_k \vert^2}{\mathbf{\widehat{v}}_k^H \left( P_k(\mathbf{\widetilde{h}}_k-\mathbf{\bar{e}}_k)(\mathbf{\widetilde{h}}_k-\mathbf{\bar{e}}_k)^H+\sum_{j \neq k}  P_j \mathbf{h}_j \mathbf{h}_j^H +N_0 \mathbf{I}\right)\mathbf{\widehat{v}}_k} \right\},
\end{eqnarray}
\setcounter{equation}{\value{mytempeqncnt}}
\hrulefill
\vspace*{4pt}
\end{figure*}
\begin{figure*}[!t]
\normalsize
\setcounter{mytempeqncnt}{\value{equation}}
\setcounter{equation}{22}
\begin{eqnarray}
\begin{split}
\mathbb{E}\lbrace \widehat{SINR}_k(d_0) \rbrace
&=\mathbb{E} \left\{{\rm tr} \left[  P_k \mathbf{\widehat{h}}_k \mathbf{\widehat{h}}_k^H \left(\mathbf{\widehat{A}}- P_k\mathbf{\widehat{h}}_k \mathbf{\widehat{h}}_k^H\right)^{-1} \right] \right\}\\
& =\mathbb{E} \left\{{\rm tr} \left[  P_k \mathbf{\widehat{h}}_k \mathbf{\widehat{h}}_k^H \left(\sum_{j \neq k}  P_j \mathbf{\widehat{h}}_j \mathbf{\widehat{h}}_j^H +(N_0+N_1+N_2) \mathbf{I} \right)^{-1} \right] \right\}.
\end{split}
\end{eqnarray}
\setcounter{equation}{\value{mytempeqncnt}}
\hrulefill
\vspace*{4pt}
\end{figure*}
\begin{figure*}[!t]
\normalsize
\setcounter{mytempeqncnt}{\value{equation}}
\setcounter{equation}{23}
\begin{eqnarray}
\mathbb{E}\lbrace \widehat{SINR}_k(d_0) \rbrace  \geq
P_k  \mathbb{E} \left\{ {\rm tr} \left[ \mathbf{\bar{h}}_k \mathbf{\bar{h}}_k^H  \left( \sum_{j \neq k}  P_j \mathbf{\widehat{h}}_j
\mathbf{\widehat{h}}_j^H +(N_0+N_1+N_2) \mathbf{I} \right)^{-1} \right] \right\} .
\end{eqnarray}
\setcounter{equation}{\value{mytempeqncnt}}
\hrulefill
\vspace*{4pt}
\end{figure*}

\section{Performance Analysis of Channel Sparsification with Imperfect CSIs}

To investigate the performance loss due to sparsity, comparing the SINR with and without channel sparsity is more intuitive than discussing the absolute value of SINR directly.

\begin{Definition}[SINR Fidelity]
The SINR fidelity compares the SINR with and without channel sparsity when link distance threshold is $d_0$, which is defined as:
\begin{eqnarray}
\rho(d_0)=\frac{\mathbb{E}\lbrace\widehat{SINR}_k(d_0)\rbrace}{\mathbb{E}\lbrace SINR_k \rbrace}.
\end{eqnarray}
\end{Definition}
Note that $\rho(d_0)$ is a ratio with value range $[0,1]$. From [19], it shows that SINR fidelity is an increasing function until reaching to maximum value with perfect CSI. To discuss the relative performance loss with imperfect CSI when channel sparsity is applied, a lower bound of SINR fidelity is derived and discussed in the following part.
\par
Given a distance threshold $d_0$, the expectation of $k$th UE's SINR in (15) becomes (20).
\setcounter{equation}{20}
By substituting (12) into (20), it can be rewritten as
\begin{eqnarray}
\mathbb{E}\lbrace \widehat{SINR}_k(d_0) \rbrace=\mathbb{E} \left\{ \frac{P_k \vert \mathbf{\widehat{v}}_k^H \mathbf{\widehat{h}}_k \vert^2}{\mathbf{\widehat{v}}_k^H \left( \mathbf{\widehat{A}}-P_k \mathbf{\widehat{h}}_k \mathbf{\widehat{h}}_k^H \right)\mathbf{\widehat{v}}_k} \right\}
\end{eqnarray}
where $\mathbf{\widehat{v}}_k=\sqrt{P_k} \mathbf{\widehat{A}}^{-1} \mathbf{\widehat{h}}_k$ is the $k$th column of detection matrix in (11). Thus (21) can be written as
\begin{eqnarray}
\begin{split}
&\mathbb{E}\lbrace \widehat{SINR}_k(d_0) \rbrace = \mathbb{E} \left\{ \frac{\sqrt{P_k} \mathbf{\widehat{h}}_k^H  \mathbf{\widehat{A}}^{-1} \mathbf{\widehat{h}}_k}{1-\sqrt{P_k}
\mathbf{\widehat{h}}_k^H  \mathbf{\widehat{A}}^{-1} \mathbf{\widehat{h}}_k} \right\} \\
&= \mathbb{E} \left\{{\rm tr} \left[ \sqrt{P_k} \mathbf{\widehat{h}}_k \mathbf{\widehat{h}}_k^H  \mathbf{\widehat{A}}^{-1} \left( \mathbf{I}-\sqrt{P_k} \mathbf{\widehat{h}}_k
\mathbf{\widehat{h}}_k^H \mathbf{A}^{-1} \right)^{-1} \right]  \right\}.
\end{split}
\end{eqnarray}

 After introducing $N_1$ and $N_2$ noted in (14), (22) can be written as (23). Next, the relatively small estimation error in
$\mathbf{\widehat{h}}_k \mathbf{\widehat{h}}_k^H$ is ignored, thus (23) turns into (24).
\par To character $N_1$ and $N_2$ in (14), a probability density function ($p.d.f$) of the link distance $d_{n,k}$ is applied here.
The distance distributions for different network area shapes are derived in [4], such as circle, square, and rectangle. The circular network area as a example is considered in this paper, whose radius is set $r$. According to these derivation results in [4], the distance distribution between two random points is:
\setcounter{equation}{24}
\begin{eqnarray}
f(x,r)=
\left\{
\begin{array}{rcl}
&\frac{r_0^2}{r^2}                                                                & x=r_0 \\
& \frac{2}{r^2}x             & r_0\leq x \leq r
\end{array} \right.,
\end{eqnarray}
where $f(x,r)$ is related to both the distance threshold and the circular area radius $r$. Note that it can be easily extended to other shape network by exchange the distance distribution derived in [22]. $\mu=\mathbb{E} \left\{ |h_{n,k}|^2\right\} =\int_{r_0}^rx^{-\alpha}f(x,r)dx$ is denoted for each $\{n,k\}$. Accordingly, the after-sparse version of $\mu$ is
$\bar{\mu}=\mathbb{E} \left\{ |\bar{h}_{n,k}|^2\right\}=\int_{r_0}^{d_0}x^{-\alpha} f(x,r) dx$.
Thus, (24) can be written as
\begin{eqnarray}
\begin{split}
&\mathbb{E}\lbrace \widehat{SINR}_k(d_0) \rbrace \geq \\
&\bar{\mu} \sqrt{P_k}  \mathbb{E} \left\{ {\rm tr}  \left( \sum_{j \neq k} \sqrt{P_j} \mathbf{\widehat{h}}_j\mathbf{\widehat{h}}_j^H +(N_0+N_1+N_2)
\mathbf{I} \right)^{-1}  \right\}.
\end{split}
\end{eqnarray}

Likewise, the expectation of SINR with no sparsification and estimation errors in (16) can be written as
\begin{eqnarray}
\begin{split}
\mathbb{E}\{SINR_k\}& = \mathbb{E} \left\{ \frac{\sqrt{P_k} \vert \mathbf{v}_k^H \mathbf{h}_k \vert^2}{\mathbf{v}_k^H \left(\sum_{j \neq k} \sqrt{P_j} \mathbf{h}_j \mathbf{h}_j^H + N_0 \mathbf{I} \right) \mathbf{v}_k} \right\}\\
&=\mu \sqrt{P_k} \mathbb{E} \left\{ \left(\sum_{j \neq k} \sqrt{P_j} \mathbf{h}_j \mathbf{h}_j^H +N_0 \mathbf{I} \right)^{-1} \right\}.
\end{split}
\end{eqnarray}
\par
Thus, the SINR fidelity can be derived by comparing $\mathbb{E}\lbrace \widehat{SINR}_k(d_0) \rbrace$ in (26) and $SINR_k$ in (27), and the corresponding result is summarized as:

\begin{theorem}[The SINR Fidelity]
Given a distance threshold $d_0$, the SINR fidelity is
\begin{eqnarray}
\rho(d_0) \geq \uline{\rho(d_0)} \triangleq \frac{\bar{\mu}}{\mu} \frac{N_0}{N_0+N_1+N_2},
\end{eqnarray}
where $N_1=\sum_{j=1}^K P_j(\mu-\bar{\mu})$ characterizes the ignored CSI by sparsification; $N_2= ( K^2/\tau P_T)\sum_{j=1}^K P_j \int_{r_0}^{d_0} f(x)dx$, $P_T$ is the transmit power of pilot, i.e. ${\rm tr}\left\{\mathbf{\Phi \Phi}^H\right\}=P_T$, characterizes the remained estimation error when UEs and RRHs located randomly.
\end{theorem}
\begin{proof}[Proof]
By ignoring the estimation error in $\sum_{j \neq k} P_S \mathbf{\widehat{h}}_j
\mathbf{\widehat{h}}_j^H$ in (26), a lower bound of SINR can be expressed as
\begin{eqnarray}
\begin{split}
&\rho(d_0)=\frac{\widehat{SINR}_k(d_0)}{SINR_k} \\
&\geq \frac{\bar{\mu}}{\mu} \frac{\mathbb{E}\left\{{\rm tr}\left(\sum_{j\neq k} P_j \mathbf{h}_j \mathbf{h}_j^H+
(N_0+N_1+N_2) \mathbf{I} \right)^{-1} \right\}}{\mathbb{E}\left\{{\rm tr}\left( \sum_{j \neq k}  P_j \mathbf{h}_j \mathbf{h}_j^H + N_0 \mathbf{I} \right)^{-1} \right\}}\\
&\geq \frac{\bar{\mu}}{\mu} \frac{\sum_{i=1}^N \frac{1}{\lambda_i+N_0+N_1+N_2}}{\sum_{i=1}^N \frac{1}{\lambda_i+N_0}},
\end{split}
\end{eqnarray}
where$ \lambda_1, \lambda_2, \dots \lambda_N$ are the eigenvalues of Hermitian semi-positive matrix $\sum_{j \neq k} P_j\mathbf{h}_j \mathbf{h}_j^H$. By multiplying a same factor $1/ N_0$ to the numerator and denominator in (29), (29) can be expressed as
\begin{eqnarray}
\begin{split}
\rho(d_0) & \geq \frac{\bar{\mu}}{\mu} \frac{\sum_{i=1}^N \frac{1}{N_0(\lambda_i+N_0+N_1+N_2)}}{\sum_{i=1}^N \frac{1}{N_0(\lambda_i+N_0)}} \\
&\geq \frac{\bar{\mu}}{\mu} \frac{\sum_{i=1}^N \frac{1}{\lambda_i+N_0} \cdot \frac{1}{N_0+N_1+N_2}}{\sum_{i=1}^N \frac{1}{\lambda_i+N_0} \cdot \frac{1}{N_0}}.
\end{split}
\end{eqnarray}
Then, the SINR fidelity is simplified to the following form:
\begin{eqnarray}
\rho(d_0) \geq \uline{\rho(d_0)} \triangleq \frac{\bar{\mu}}{\mu} \frac{N_0}{N_0+N_1+N_2}.
\end{eqnarray}

Given a $p.d.f$ of the link distance mentioned in (25), $N_1$ becomes
\begin{eqnarray}
N_1=\mathbb{E} \left\{ \sum_{j \neq k}P_j  \vert \widetilde{h}_{nj} \vert^2 \right\}= \sum_{j \neq k} P_j (\mu-\bar{\mu}),
\end{eqnarray}
which declines with $d_0$. When the distance threshold satisfies $d_0=r$, i.e., no sparsification is implied, $N_1$ will vanish.
\par
And $N_2$ can be expressed as
\begin{eqnarray}
N_2= \mathbb{E} \left\{  \sum_{j \neq k} P_j \vert \bar{e}_{nj} \vert^2 \right\}=\frac{K^2}{\tau P_T}  \sum_{j \neq k} P_j \int_{r_0}^{d_0} f(x)dx,
\end{eqnarray}
which increases with the increase of $d_0$, and reaches the maximum value when $d_0=r$.
\end{proof}
\par

As shown in (28), SINR fidelity increases with transmit power $P_k$, it also indicates that performance loss due to channel sparsity increases with the total transmit power $P_{total}$.
Note that estimation method we use has no effect on the conclusions of our work, except the specific coefficient variation of $N_2$:
\begin{eqnarray}
N_2^{MMSE}= \mathbb{E} \left\{ \sum_{j \neq k} P_j \vert \bar{e}_{nj} \vert^2 \right\} = \left(\frac{\mu}{\mu+N_0}\right)^2 N_2^{LS}.
\end{eqnarray}
To be clear, we have provided the form and corresponding derivation of $N_2$ in Appendix A when MMSE estimator is applied.

\begin{theorem}
A sole optimal distance threshold that maximizes the SINR fidelity exists when RRHs and UEs located randomly.
\begin{proof}
Please refer to Appendix B
\end{proof}
\end{theorem}
\par
 $\rho(d_0)$ is a convex function of distance threshold $d_0$, after the first increase benefit from cooperation, it declines due to introducing too many harmful channels with little channel gains but considerable estimation error. The distance threshold maximizing the SINR fidelity is the optimal $d_0$ that can balance the tradeoff between cost and gain.\par
By introducing the $p.d.f$ of link distance in (26), the optimal distance threshold $d_{opt}$ is the positive real number solution of $\dot{\rho}(d_0)=0$, which means to solve the following equation:
\begin{eqnarray}
\begin{split}
&(2-\alpha)A(2N_0+2B r^{2-\alpha}-\alpha B r_0^{2-\alpha})d_0^{1-\alpha}\\
&+2\alpha AC (r_0^{2-\alpha}d_0-d_0^{3-\alpha})=0,
\end{split}
\end{eqnarray}
where $A$, $B$, and $C$ are the defined in Appendix A. Unfortunately, (35) does not have a closed-form.

\section{Optimization of Channel Matrix Sparsity with Dinkelbach-based Algorithm}

The results of Section III are built on the assumption that all RRHs and UEs are independent and identically distributed. In this part, an alternating-Dinkelbach optimization is proposed to extend the study to a more general case, in which the optimal distance threshold can be obtained as long as the specific $p.d.f$ is known.
All the required parameters in the proposed algorithm can be obtained in the practical systems, the computational complexity is acceptable and can be implemented. Note that the proposed algorithm is with channel estimation, which cannot be avoided in practical systems, and thus the proposed algorithm is more applicable.
\par
Based on (29), obtaining the optimal $d_0$ is equivalent to solving the following optimization problem:
\begin{eqnarray}
(P1):\arg\max_{d_0}=\frac{F_1(d_0)}{F_2(d_0)},
\end{eqnarray}
where
\begin{eqnarray*}
F_1(d_0)=N_0 \int_{r_0}^{d_0}x^{-\alpha}f(x)dx>0,
\end{eqnarray*}
and
\begin{eqnarray*}
\begin{split}
F_2(d_0)= &\int_{r_0}^rf(x)dx \left[N_0+ \sum_{j=1}^K P_j \int_{d_0}^r x^{-\alpha} f(x)dx \right. \\
&\left. + \frac{K^2}{\tau P_T }  \sum_{j=1}^K P_j \int_{r_0}^{d_0}f(x)dx \right]>0.
\end{split}
\end{eqnarray*}

The objective function in Problem (P1) can be transferred from the fractional form to the subtractive form via Dinkelbach algorithm. In particular,
defining a parameter $q$ as
\begin{eqnarray}
q=\frac{F_1(d_0)}{F_2(d_0)}.
\end{eqnarray}

Then, a subtractive form optimization problem with a given parameter $q$ can be formulated as
\begin{eqnarray}
(P2):\arg\max_{d_0} F_1(d_0)-q F_2(d_0).
\end{eqnarray}

To solve Problem (P2), we first present the following lemma.\\
\begin{lemma} \label{1}
\begin{eqnarray}\nonumber
q'=\frac{F_1(d_0')}{F_2(d_0')}= \max_{d_0} \frac{F_1(d_0)}{F_2(d_0)},
\end{eqnarray}
if and only if\
\begin{eqnarray}
\begin{split}
F(q')=\max_{d_0} F_1(d_0)-q' F_2(d_0) \\
=F_1(d_0')-q 'F_2(d_0') =0.
\end{split}
\end{eqnarray}
\end{lemma}

\par

\begin{proof}[Proof]
In(39), if $d_0 \neq d_0'$,
\begin{eqnarray}\nonumber
\frac{F_1(d_0)}{F_2(d_0)} < \frac{F_1(d_0')}{F_2(d_0')}= q',
\end{eqnarray}
then leads to
\begin{eqnarray}\nonumber
F_1(d_0)-q'F_2(d_0)<0,
\end{eqnarray}
which is in conflict with the condition in (39).
\end{proof}

 $\mathbf{Lemma}$ \ref{1} reveals to solve Problem (P1) with an objective function in fractional form, there exists a corresponding Problem (P2) in
subtractive form. Moreover, $\mathbf{Lemma}$ \ref{1} provides the condition when two problem formulations can lead to the same optimal solution $d_0'$.
To reach the condition in (39), we focus on Problem (P2) first.
\begin{lemma}
The optimization objective in Problem (P2) is a concave function in terms of $d_0$.
\end{lemma}

\begin{proof}
Please refer to Appendix 3
\end{proof}

\begin{figure*}[!t]
\normalsize
\setcounter{mytempeqncnt}{\value{equation}}
\setcounter{equation}{46}
\begin{eqnarray}
\rho(d_0)=\frac{\bar{\mu}}{\mu} \frac{N_0}{N_0+PK(\mu-\bar{\mu})+P \left[ 2K\left(1-\frac{\tau}{K} \right) \mu + \frac{K^3}{\tau P_T} \right] \int_{r_0}^{d_0} f(x)dx }.
\end{eqnarray}
\setcounter{equation}{\value{mytempeqncnt}}
\hrulefill
\vspace*{4pt}
\end{figure*}
 According to Lemma 4 and the fact that the feasible set for $d_0$ is convex, Problem (P2) is convex optimization problem. With the bisection method, the optimum solution of (P2) can be obtained.
\begin{algorithm}
\caption{Dinkelbach Algorithm of Distance Threshold $d_0$}
\begin{algorithmic}[1]
\STATE \textbf{Initialization:}\\
Set the parameter as $q=0$, the index of iteration as $n=1$, and the judgement of convergence as $conv=0$.\\
Set the maximum number of iterations as $n_{max}$, and the threshold of termination as $\Delta$, which is a constant that approaches 0.\\
\STATE \textbf{Repeat:}\\
Set $n=n+1$,\\
Solve Problem (P2) through bisection method, and mark the optimal solution as $d_0^n$,\\
  \textbf{If} $F(q)=F_1(d_0^n)-qF_2(d_0^n)<\Delta$\\
Mark the optimal solution as $d_0'=d_0^n$, and set $conv=1$.\\
  \textbf{Else}
Update $q$ according to (37).\\
\STATE \textbf{Until:}
$conv=1$ or $n=n_{max}$.\\
\STATE \textbf{Return:} the optimal distance threshold $d_{opt}$.\\
\end{algorithmic}
\end{algorithm}
\begin{lemma}
$\mathbf{Algorithm}$ 1 converges to $q'$ and $d_0'$, which satisfy (39) in $\mathbf{Lemma}$ 3.
\end{lemma}
\par
 \begin{proof}
 For the purpose of explanation, at the $n$th iteration, when denoted by $q_n$, the optimized solution to Problem (P1) is denoted by $d_0^n$. According
to (37), the parameter is updated by $q_{n+1}=\frac{F_1(d_0^n)}{F_2(d_0^n)}$ in the next iteration.

First, it is revealed that the optimized value $F(q)$ in Problem (P2) is non-negative, i.e.,

\begin{eqnarray}
\begin{split}
F(q_{n+1})& =\max_{d_0} F_1(d_0)-q_{n+1}F_2(d_0) \\
& \geq F_1(d_0^n)-q_{n+1}F_2(d_0^n)=0.
\end{split}
\end{eqnarray}

Next, it is revealed that the parameter $q$ increases after each iteration. According to $\mathbf{Lemma}$ 3 and (37), since the iteration process is not terminated in the $(n + 1)$th iteration, it can be determined that $F(q_n) > 0$ and $F(q_{n+1}) > 0$. Thus, $F(q_n)$ can be expressed as

\begin{eqnarray}
\begin{split}
F(q_n)&=F_1(d_0^n)-q_n F_2(d_0^n) \\
&=(q_{n+1}-q_n) F_2(d_0^n) >0.
\end{split}
\end{eqnarray}

Further, it is shown that $F(q)$ decreases after each iteration. Due to the increasing of $q$ after each iteration, $F(q_n)$ can be expressed as
\begin{eqnarray}
\begin{split}
F(q_n)&=\max_{d_0} F_1(d_0)-q_n F_2(d_0) \\
&\geq F_1(d_0^{n+1})-q_n F_2(d_0^{n+1}) \\
&> F_1(d_0^{n+1})-q_{n+1}F_2(d_0^{n+1})\\
&=F(q_{n+1}).
\end{split}
\end{eqnarray}
\par
Based on the properties in (40) and (42), $F(q_{n})$ is proved to be a nonnegative value after each iteration. As the proposed algorithm indicated, after sufficient iterations (less than $n_{max}$), $F(q)$ can approach to 0 (gap is less then the predefined nonnegative threshold of termination $\Delta$), and the corresponding $d_0^n$ is believed to be close enough to optimal $d_0$.

 As a result, $\mathbf{Algorithm}$ 1 converges
to $d_0'$ and $q'$, which satisfy (39) in $\mathbf{Lemma}$ 3. Therefore, $d_0'$ is the optimal solution to Problem (P1).
\end{proof}


The computational complexity of Algorithm 1 is reflected to the number of iterations, which can be estimated as $\mathcal{O} ( log(1/\Delta))$ [13], and $\Delta$ is a prescribed accuracy. In each iteration, problem (P2) is solved by using bisection method, which is in order of $\mathcal{O} ( log(r))$ [14], and $r$ is the system size. Therefore, the computation complexity of our Algorithm 1 is in order of $\mathcal{O} ( log(r)log(1/\Delta))$.

\section{System Performance with Non-orthogonal Training}

In large C-RANs, as the orthogonal training would need to be least $K$ symbols long, which is infeasible for large $K$, non-orthogonal training sequences could be utilized. In particular, the short channel coherence time due to mobility does not allow for such long training sequences.

\subsection{Channel Estimation Phase}

The extent to which pilot contamination impacting on the system performance under different scenarios have been researched in the literatures [23] and [24]. During analyzing the effects of pilot contamination, the worst-case scenario is considered in this paper, in which no allocation of training sequences or subspace estimation technique is applied. The training sequences are generated by the discrete Fourier transform matrix in [25] with $\tau < K$ , thus, $2(K-\tau)$ UEs are interfered by other UEs due to the properties of training matrix $\mathbf{\Phi}$ when the training length is $\tau$. Hence, the probability of one UE suffering from other UE's interference becomes $2(1-\frac{\tau}{K})$. As a result, the interference due to the pilot contamination of random UE is quantified as the probability of suffering from the pilot contamination multiplying by the average magnitude of random channel coefficient. Therefore, the observed channel coefficient can be expressed as:
\begin{eqnarray}
\widehat{h}_{nk}=h_{nk}+\epsilon_{nk}+e_{nk},
\end{eqnarray}
where $\epsilon_{n,k} \sim \mathcal{CN}\left(0,2(1-\frac{\tau}{K})\mathbb{E}\{\vert h_{n,k} \vert^2\} \right)$ represents the interference from other channels;
$e_{n,k} \sim \mathcal{CN}\left(0,\frac{K^2}{\tau P_T} \right)$ represents the Gaussian noise; $\mathbb{E}\{\vert h_{n,k} \vert^2\}$ is the statistic average
of squared amplitude of random channel coefficient. To obtain the statistic characters of $\epsilon_{n,k}$, the $p.d.f$ of the link distance can be used. With $\mu=\mathbb{E}\{\vert h_{n,k} \vert^2\}$, the interference can be expressed as
\begin{eqnarray}
 \epsilon_{n,k} \sim \mathcal{CN}\left(0, 2\left(1-\frac{\tau}{K}\right)  \mu \right).
\end{eqnarray}

Thus, the expectation of covariance of estimation error using non-orthogonal training becomes
\begin{eqnarray}
\mathbb{E} \left\{  \mathbf{EE}^H \right\} = \left[ 2K\left(1-\frac{\tau}{K} \right) \mu + \frac{K^3}{\tau P_T} \right]
\mathbf{I}.
\end{eqnarray}

Same as $N_2$ in (33), $N_2$ with non-orthogonal training can be written as
\begin{eqnarray}
N_2= P \left[ 2K\left(1-\frac{\tau}{K} \right) \mu + \frac{K^3}{\tau P_T} \right] \int_{r_0}^{d_0} f(x)dx.
\end{eqnarray}
\par
 $N_1$ and $N_2$ both increase with the decreasing $\tau$. However, the selection of optimal training length is determined by the tradeoff between the
quality of channel estimate and the information throughput. Hence, the channel accesses employed for pilot transmission and for data transmission needs to be
optimized for maximizing total throughput and fairness of the system.

\subsection{SINR Fidelity with Non-orthogonal Training}

Except the different form of $N_2$ we proposed in former part, the derivation of SINR fidelity with non-orthogonal training is similar
with the derivation with orthogonal training in Part A, Section III. Based on (32), (33) and (46), the lower bound of SINR fidelity with non-orthogonal training is expressed in (47).
To prove the convexity of SINR fidelity with non-orthogonal training, the similar derivations of (47) have been made by applying the $p.d.f$ of link distance given in (25). (Proof: See Appendix B)
\par
Comparing (47) and (28),
it can be concluded that the severe interference caused by pilot contamination grows with the service radius of RRHs increasing. The influence of interference caused by pilot contamination on the SINR fidelity is similar with that of estimation error, but the former one is more serious. In the future, the training schedule should be done to further decrease the pilot interference.

\section{Numerical Results}

To verify the analysis results for optimal distance threshold with imperfect CSIs in large C-RANs, several simulation results under different scenarios are presented in this section.
 UEs and RRHs are uniformly scattered in a circular area with $r=5\,\rm{km}$, the minimum distance between RRH and UE is $10\,\rm{m}$. The path loss exponent is $3.8$. For each realization, $1000$ iterations are simulated with independent channel state. MATLAB is used in this paper to simulate and verify our research results. If not specifically pointed out, the simulation parameters are summarized in TABLE I.
\begin{table*}[!htp]
\newcommand{\tabincell}[2]{\begin{tabular}{@{}#1@{}}#2\end{tabular}}
\caption{SIMULATION PARAMETERS SETTINGS}
\centering
\small
\begin{tabular}{|l|l|l|}
\hline
Symbol                 & Definition           &Value\\
\hline
$N$          & The number of RRHs                         &1000\\
\hline
$K$          & The number of UEs                          &800\\
\hline
$\alpha$       &Path loss exponent                          &3.8\\
\hline
$r_0$          &Minimum distance between RRHs and UEs       &10 $m$\\
\hline
$P_k$      &Transmit power of data block                &23 $dBm$\\
\hline
$N_0$      &Noise power spectral density               &-174 $dBm/Hz$\\
\hline
$\Delta$        &Threshold of termination                &$10^{-4}$ \\
\hline
$n_{max}$       &Maximum number of iteration               &$20$ \\
\hline
\end{tabular}
\end{table*}
\par

\begin{figure}[t!]\centering
\label{sys}
\centering
\includegraphics[width=0.43\textwidth]{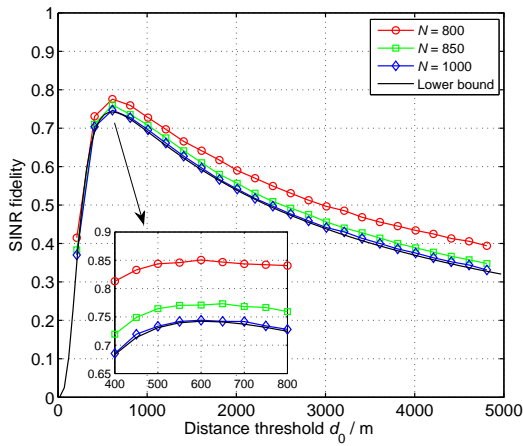}
\vspace*{-10pt} \caption{SINR fidelity vs the distance threshold and the lower bound}\vspace*{-10pt} \label{sys}
\end{figure}

\begin{figure}[t!]\centering
\label{sys}
\centering
\includegraphics[width=0.43\textwidth]{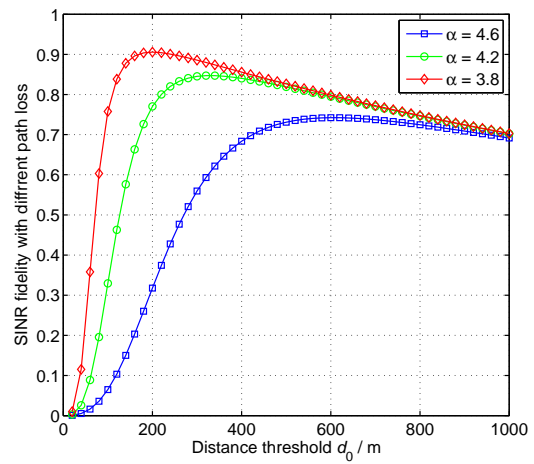}
\vspace*{-10pt} \caption{SINR fidelity vs distance threshold with different path loss exponent}\vspace*{-10pt} \label{sys}
\end{figure}

\begin{figure}[t!]\centering
\label{sys}
\centering
\includegraphics[width=0.43\textwidth]{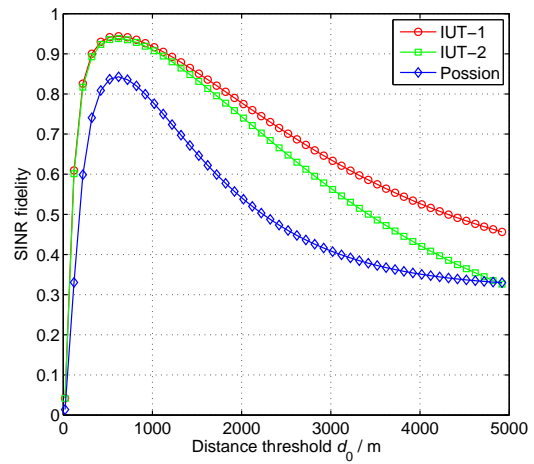}
\vspace*{-10pt} \caption{SINR fidelity vs distance threshold with different $p.d.f$ of link distance}\vspace*{-10pt} \label{sys}
\end{figure}

\begin{figure}[t!]\centering
\label{sys}
\centering
\includegraphics[width=0.43\textwidth]{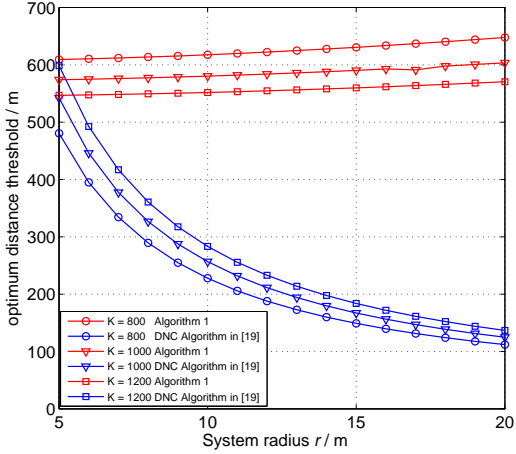}
\vspace*{-10pt} \caption{Two ways of obtaining optimum distance threshold }\vspace*{-10pt} \label{sys}
\end{figure}
First, the SINR fidelity and corresponding lower bounds with imperfect CSIs are illustrated. As shown in Fig.\;2, the derived lower bounds are coincide with the simulation
results throughout the whole $d_0$ region under different user density, which validates the high accuracy of our analysis. Since Monte Carlo simulation results vary with $N$, while the lower bounds remain unchanged with different $N$, we set the number of RRHs to be $N=800\,/\,850\,/\,1000$.  Fig. 2 illustrates the derived SINR fidelity is more matchable with Monte Carlo results when $N$ is bigger, which indicates this research is more recommended to applied in dense networks. However, to be clear, $N$ does not have to be greater than $K$, since the convexity of SINR fidelity still holds with a smaller curvature when $N \leq K$. \par

As can be seen, the results of Monte Carlo simulation are coincide with the lower bounds.
The optimal distance threshold $d_0$ converges to a constant when the number of RRHs goes to infinity. This feature implies that the number of nonzero entries per row or per column in $\widehat{H}$ (which is approximately $d_0^2/r^2$) does not scale with the number of RRHs in a large dense C-RANs. This enlighten us assume that each RRH only needs to estimate the CSIs of a small number of users that are close to it when RRHs are deployed in a denser way comparing to UEs.\par

The LOS and NLOS characteristics would be different in practice, which indicates the path loss exponent would experience a sudden raise as the increasing link distance. Fig. 3 illustrates the SINR fidelity with different path loss exponent and it shows that the optimal distance threshold decreases when the system experiences a deep path loss. As such, the convexity of SINR fidelity is guaranteed in this case.
\par
Next, the SINR fidelity with different $p.d.f$ of link distance is displayed in Fig. 4, where the link distance of IUT-1 follows independently and uniformly distribution, corresponding $p.d.f$ is as follows:
\begin{eqnarray*}
\begin{split}
&f(x,r)=\\
&\left\{
\begin{array}{rcl}
&\int_0^{d_0} \frac{2x}{r^2} \left( \frac{2}{\pi} arccos\left(\frac{y}{2r}\right)-\frac{y}{2r}\sqrt{1-\frac{y^2}{4r^2}}\right)                                                                & x=r_0 \\
& \frac{2x}{r} \left(\frac{2}{\pi} arccos\left( \frac{x}{2r}\right)-\frac{x}{2r}\sqrt{1-\frac{x^2}{4r^2}}\right)             & r_0\leq x \leq r
\end{array} \right.,
\end{split}
\end{eqnarray*}
 and IUT-2 is an approximation of IUT-1 when system size $r$ is sufficiently large, the corresponding $p.d.f$ is noted in (25). The $p.d.f$s of Poisson distribution is
 \begin{eqnarray*}
f(x,n)=2(\pi \lambda) x e^{-\pi \lambda x^2}, r>0,
\end{eqnarray*}
where $\lambda$ is the density of Poisson Point Process, which is set to be $1/\pi r^2$ in our simulations.
It can be seen that the convexity of SINR fidelity holds when RRHs and UEs follow different distributions.
\par

To make it clear, the difference between the Dinkelbach-based algorithm we proposed and the DNC algorithm we mentioned in the introduction section has to be emphasized. The DNC algorithm is used in the progressing of received signals to reduce the computation complexity of large-scale cooperative system. In the DNC algorithm, distance threshold $d_0$ is given as
\setcounter{mytempeqncnt}{\value{equation}}
\setcounter{equation}{47}
\begin{eqnarray}
d_0=\left( r^{2-\alpha}+\frac{\alpha r_0^{2-\alpha}-2r^{2-\alpha}(1-\rho')N_0}{2N_0+\frac{2\rho'(\alpha r_0^{2-\alpha}-2r^{2-\alpha})(K-1)P_k}{(\alpha-2)r^2}} \right)^{\frac{1}{\alpha-2}},
\end{eqnarray}
when the SINR requirement $\rho'$ is previously determined and channel estimation has not been considered. Specifically, these two ways of obtaining distance threshold in the case of the number of UEs $K = 800$ are compared in Fig. 5 and Fig. 6, respectively. As shown in Fig. 5, the obtained distance threshold in DNC algorithm is smaller than that in the proposed algorithm. In Fig. 6, the optimal SINR fidelity obtained by the proposed algorithm is better than that of the DNC algorithm in [19], which indicates that the utilization of Dinkelbach-based algorithm will lead to superior performance to DNC algorithm with the existent of channel estimation error.

\par
\begin{figure}[t!]\centering
\label{sys}
\centering
\includegraphics[width=0.43\textwidth]{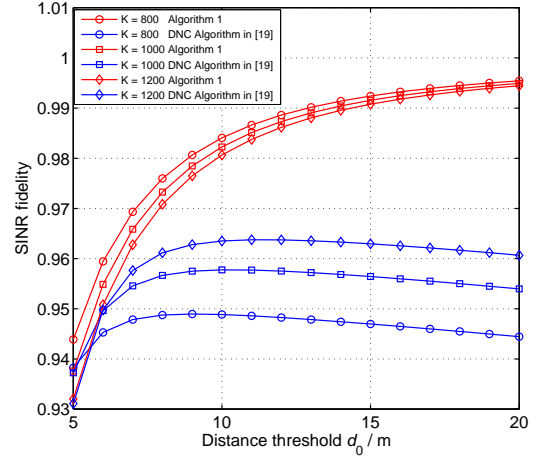}
\vspace*{-10pt} \caption{Maximum SINR fidelity obtained by Dinklebach-based Algorithm and DNC Algorithm in [19]}\vspace*{-10pt} \label{sys}
\end{figure}

\begin{figure}[t!]\centering
\label{sys}
\centering
\includegraphics[width=0.43\textwidth]{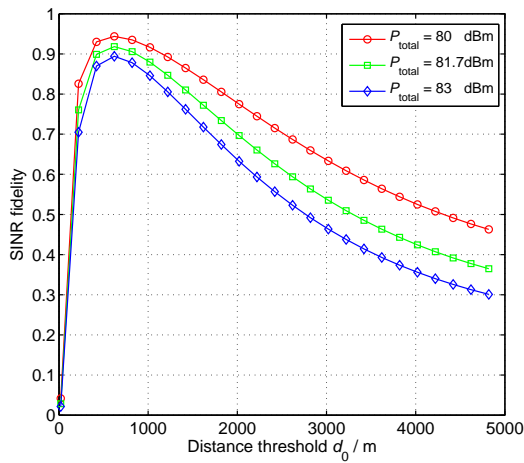}
\vspace*{-10pt} \caption{SINR fidelity vs distance threshold with different total transmit power }\vspace*{-10pt} \label{sys}
\end{figure}

\begin{figure}[t!]\centering
\label{sys}
\centering
\includegraphics[width=0.43\textwidth]{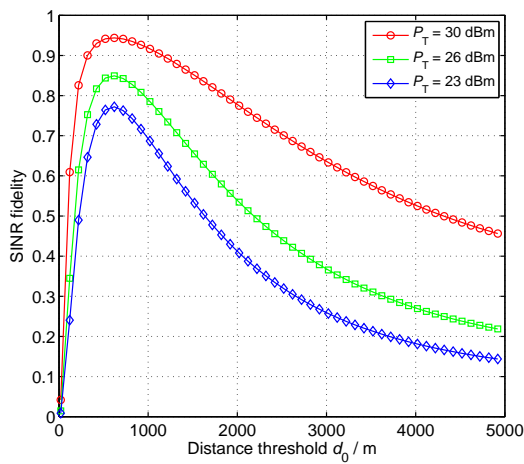}
\vspace*{-10pt} \caption{SINR fidelity with different transmit power of pilot when transmit power of data is fixed to 23 $dBm$}\vspace*{-10pt} \label{sys}
\end{figure}

\begin{figure}[t!]\centering
\label{sys}
\centering
\includegraphics[width=0.43\textwidth]{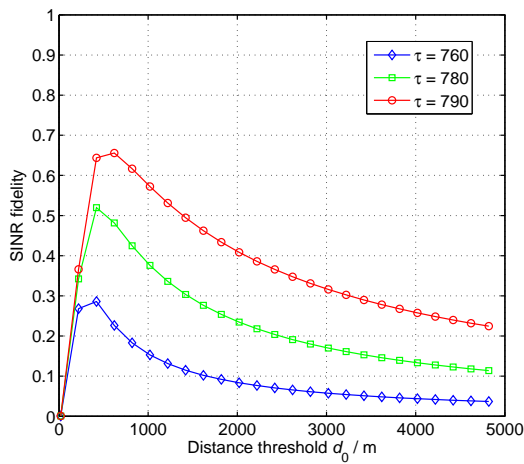}
\vspace*{-10pt} \caption{SINR fidelity with different training length $\tau=760 \;/\;780\; /\; 790$ when the number of user is $K=800$}\vspace*{-10pt} \label{sys}
\end{figure}

In Fig. 7, the SINR fidelity with different total transmit power $P_{total}$ has been illustrated, SINR fidelity decreases with $P_{total}$ as shown. The effect of transmit power of training block is illustrated in Fig. 8, where the transmit power of data block is fixed to 23 dBm, and the transmit power of pilot block is set to 23 dBm, 26 dBm, 30 dBm, respectively. With the decrease of transmit power of pilot block, the influence due to estimation error increases, which leads to the degradation of SINR fidelity. The simulation results of Fig. 8 suggest that proposal still work efficiently and the corresponding results still hold when the transmit power of pilot is near or even equals to the transmit power of data.
\par
\begin{figure*}[!t]
\normalsize
\setcounter{mytempeqncnt}{\value{equation}}
\setcounter{equation}{54}
\begin{eqnarray}
\uline{\dot{\rho}(d_0)}=\frac{(2-\alpha)A(2N_0+2B r^{2-\alpha}-\alpha B r_0^{2-\alpha})d_0^{1-\alpha}+2\alpha AC (r_0^{2-\alpha}d_0-d_0^{3-\alpha})}
{(N_0+Br^{2-\alpha}-Bd_0^{2-\alpha}+C d_0^2)^2},
\end{eqnarray}
\setcounter{equation}{\value{mytempeqncnt}}
\hrulefill
\vspace*{4pt}
\end{figure*}
Fig.\;9 illustrates the effects of non-orthogonal training on system performance when the number of UEs is $800$. The trend of
SINR fidelity with non-orthogonal training is similar to which with orthogonal training, which is in conform with our analysis in Part B,
Section V. With a slightly decrease of training length, the SINR fidelity deteriorate badly. The severe performance degradation caused by pilot contamination still remains as a big issue for large C-RANs.

\section{Conclusion}

In this paper, an analytical model has been presented to character the performance of large C-RANs with sparse channel matrix and estimation error.
To explore the performance loss in terms of the distance threshold with imperfect CSI, a lower bound of a ratio comparing the SINR with and without sparsification of channel matrix has been derived. This ratio has been proved to be a convex function of service distance with both orthogonal and non-orthogonal training sequences when RRHs and UEs are distributed randomly. An algorithm has been presented to obtain the optimal distance threshold when RRHs and UEs distribute unevenly. Note that power control strategy has a great impact on the performance of channel sparsity in large C-RANs [26]. For one UE, the increase of transmit power leads to a better SINR fidelity, however, the joint optimization of all UEs' SINR fidelity remains a challenging issue, which will be researched in the future.

\appendices
\section{Derivation of $N_2$ when MMSE estimator is applied}

The MMSE estimation of channel matrix is:
\begin{eqnarray}
\mathbf{H}^{MMSE} = \mathbf{S}(\mathbf{\Phi}^H \mathbf{R}_H \mathbf{\Phi} + N_0  \mathbf{I})^{-1}\mathbf{\Phi}^H \mathbf{R}_H,
\end{eqnarray}
where $\mathbf{R}_H = \mathbb{E} \{ \mathbf{H}^H \mathbf{H} \}$ is the correlation matrix of channel matrix $\mathbf{H}$.
\par
Since the fast fading coefficients of all channels are considered to be uncorrelated, the correlation matrix can be expressed as:
\begin{eqnarray}
\begin{split}
\mathbf{R}_H &= \mathbb{E} \{ \mathbf{H}^H \mathbf{H} \}=\mathbb{E}\left\{ diag \left[ \sum_{j=1}^N \vert h_{j1} \vert^2, \dots, \sum_{j=1}^N \vert h_{jK} \vert^2 \right]\right\}\\
&= \mathbb{E}\left\{ diag \left[ \sum_{j=1}^N  d_{j1}^2, \dots, \sum_{j=1}^N d_{jK}^2 \right] \right\}.
\end{split}
\end{eqnarray}
Then the estimates of channel matrix can be expressed as:
\begin{eqnarray}
\mathbf{H}^{MMSE} =  \mathbf{S} \mathbf{\Lambda}^{-1} \mathbf{\Phi}^H \mathbf{R}_H
= \mathbf{H} \mathbf{\Lambda}^{-1} \mathbf{R}_H +\mathbf{N} \mathbf{\Lambda} \mathbf{\Phi} ^H \mathbf{R}_H ,
\end{eqnarray}
where $\mathbf{\Lambda} = diag \left[ \sum_{j=1}^N  d_{j1}^2+  N_0, \dots, \sum_{j=1}^N  d_{jK}^2+ N_0 \right]$. Since $\sum_{j=1}^N  d_{jk}^2 \gg N_0$, (51) can be expressed as:
\begin{eqnarray}
\mathbf{H}^{MMSE} = \mathbf{H} +\mathbf{N} \mathbf{\Lambda}^{-1} \mathbf{\Phi} ^H \mathbf{R}_H=\mathbf{H+E}.
\end{eqnarray}
After the sparsification, the sparsed estimation matrix is:
\begin{eqnarray}
\bar{\mathbf{E}}^{MMSE} = \mathbf{\Lambda}^{-1} \mathbf{R}_H \bar{\mathbf{E}}^{LS}
\end{eqnarray}
and the $N_2$ of MMSE estimator can be expressed as:
\begin{eqnarray}
N_2^{MMSE}= \mathbb{E} \left\{ \sum_{j=1}^K P_j \vert \bar{e}_{nj} \vert^2 \right\} = \left(\frac{\mu}{\mu+N_0}\right)^2 N_2^{LS}.
\end{eqnarray}
From (54), it is proved that the form of $N_2$ suffers a coefficient variation when different estimator is applied, and the performance of MMSE estimator is better than the LS estimator.However, in general, the convexity of SINR fidelity is still holds and the proposed algorithm is still reliable when MMSE estimator is applied.

\section{Proof of the convexity of SINR fidelity with orthogonal training}

By introducing the $p.d.f$ of link distance in (26), $\uline{\rho(d_0)}$ can be written as
\setcounter{equation}{55}
\begin{eqnarray}
\begin{split}
\uline{\rho(d_0)} = &\frac{2(d_0^{2-\alpha}-r_0^{2-\alpha})+(2-\alpha)r_0^{2-\alpha}}{2(r^{2-\alpha}-r_0^{2-\alpha})+(2-\alpha)r_0^{2-\alpha}}\\
& \cdot \frac{N_0}{N_0+\frac{2P_S K}{(2-\alpha)r^2}(r^{2-\alpha}-d_0^{2-\alpha})+\frac{P_S K^2}{P_T}\frac{d_0^2}{r^2}},
\end{split}
\end{eqnarray}
where we let $\tau=K$ here to achieve the shortest training length. \par
To analyze the curve trends of $\uline{\rho(d_0)}$, it will be insightful
to calculate its derivations.
the first order derivatives of
$\uline{\rho(d_0)}$ can be expressed in (55),
\setcounter{equation}{56}
where
\begin{eqnarray}
\begin{split}
& A=\frac{N_0}{2(r^{2-\alpha}-r_0^{2-\alpha})+(2-\alpha)r_0^{2-\alpha}}, \\
&B=\frac{2P_S K}{(2-\alpha)r^2},\\
&C=\frac{P_S K^2}{P_T r^2},
\end{split}
\end{eqnarray}
respectively. Since we only concern
 about the trends of $\uline{\rho(d_0)}$, it is efficient to consider the derivatives of numerator in (55), denoted as
  $ \dot{z}(d_0)$, which can be calculated and expressed as
\begin{eqnarray}
\begin{split}
\dot{z}(d_0)=&(2-\alpha)(1-\alpha)\\
&\cdot A(2N_0+2B r^{2-\alpha} -\alpha B r_0^{2-\alpha})d_0^{-\alpha}\\
& -2\alpha(3-\alpha) AC d_0^{2-\alpha} +2\alpha A C r_0^{2-\alpha}.
\end{split}
\end{eqnarray}

Note that $A < 0 , B <0$, and $C>0$, then $\dot{z}(d_0) < 0$ for all $d_0 \in (r_0,r)$. Meanwhile
\begin{eqnarray}
\begin{split}
z(r_0)=&(2-\alpha)A(2N_0+2B r^{2-\alpha}-\alpha B r_0^{2-\alpha})r_0^{1-\alpha}>0, \\
z(r)=&(2-\alpha)A(2N_0+2B r^{2-\alpha}-\alpha B r_0^{2-\alpha})r^{1-\alpha}\\
&+2\alpha AC r(r_0^{2-\alpha}-r^{2-\alpha})<0.
\end{split}
\end{eqnarray}

Since $\dot{z}(d_0)$ is constantly negative, $z(d_0)$ is a monotonically
decreasing function in terms of $d_0$. Based on (58), it can be deduced that there will be one and only one zero point in
$z(d_0)$ while $d_0 \in (r_0,r]$. Denote the corresponding zero point as $d_{opt}$,
since the denominator of $\uline{\dot{\rho}(d_0)}$ is always positive, $\uline{\rho(d_0)}$ will reach the maximum when $d_0=d_{opt}$.
\section{Proof of the convexity of SINR fidelity with non-orthogonal training}
Based on (46), by introducing the $p.d.f$ in (26), the only difference
is $B$ in (56), when the training is non-orthogonal, it turns into
\begin{eqnarray}
C=\left[\frac{P_S K^2}{P_T}+2P_S K \left(1-\frac{\tau}{K}\right)\mu\right]/r^2.
\end{eqnarray}

It is clear that $C>0$, then (53) still holds with new $N_2$, so it can be deduced that (54) is also a convex function when UEs and RRHs are located
randomly, $\mathbf{Theorem}$ 2 remains true when the training is non-orthogonal. Since the convexity of (54) is also cannot be ensured when the formation
of $p.d.f$ is uncertain,
the Dinkelbach Algorithm we proposed before is also applicable in this case.

\section{Proof of Lemma 3}

To prove the optimization objective in Problem (P2) is a concave function in terms of $d_0$, we have to prove the second derivative of the optimization objective is negative. Let $G(d_0)=F_1(d_0)-qF_2(d_0)$, the second derivative of $G(d_0)$ is
\begin{eqnarray}
G''(d_0)=F_1''(d_0)-q F_2''(d_0),
\end{eqnarray}
where
\begin{eqnarray}
F_1''(d_0)=-\alpha N_0 d_0^{-\alpha-1}f'(d_0),
\end{eqnarray}
and
\begin{eqnarray}
\begin{split}
&F_2''(d_0)\\
&=\int_{r_0}^r f(x)dx \sum_{j=1}^K P_j \left[ \alpha d_0^{-\alpha-1} f'(d_0) +\frac{K^2}{\tau P_T} f'(d_0) \right].
\end{split}
\end{eqnarray}
It can be found that $F_1''(d_0)<0$, $F_2''(d_0)>0$, and $q \in (0,1)$, so $G''(d_0)$ will always be negative, the optimization objective in Problem (P2) is a concave function in terms of $d_0$.

\end{document}